\documentclass{eptcs}
\providecommand{\event}{LFMTP 2010}
\usepackage{breakurl}

\usepackage[latin1]{inputenc}
\usepackage{url}

\usepackage{amsmath,amssymb,amsthm}
\usepackage[english]{babel}
\usepackage{cdsty}

\title{Pattern Unification for the Lambda Calculus with Linear and
Affine Types\thanks {\footnotesize This work was
   in part supported by NABIIT grant 2106-07-0019 of the Danish Strategic Research
   Council.}}

\author{Anders Schack-Nielsen
   \qquad\qquad
   Carsten Sch\"urmann
\institute{ 
  IT University of Copenhagen\\
  Copenhagen, Denmark \\
  \texttt{anderssn|carsten@itu.dk}}
  }
\def\titlerunning{Pattern Unification for the Lambda Calculus with Linear and
Affine Types}
\def\authorrunning{A. Schack-Nielsen \& C. Sch\"urmann}

\newcommand{\aand}{\;\&\;}
\newcommand{\lhat}[1]{\widehat{\;#1\;}}
\newcommand\affarr{\mathrel{-}\joinrel\mathrel{@}}
\newcommand\afflam{\mathring{\lambda}}
\newcommand\affapp{{@}}
\newcommand\affext{\mathring{;}\,}
\newcommand\fI{\textup{$\mathbf{I}$}}
\newcommand\fA{\textup{$\mathbf{A}$}}
\newcommand\fL{\textup{$\mathbf{L}$}}
\newcommand\fUA{{\textup{\textbf{U}}_\textup{\textbf{A}}}}
\newcommand\fUL{{\textup{\textbf{U}}_\textup{\textbf{L}}}}
\newcommand{\affweak}{\succ_\mathrm{aff}}

\newtheorem{thm}{Theorem}[section]
\newtheorem{lem}[thm]{Lemma}
\theoremstyle{definition}
\newtheorem{defin}[thm]{Definition}

\begin{document}
\maketitle
\begin{abstract}
  We define the pattern fragment for higher-order unification problems
  in linear and affine type theory and give a deterministic
  unification algorithm  that computes most general
  unifiers.
\end{abstract}

\section{Introduction}
Logic programming languages, type inference algorithms, and automated
theorem provers are all examples of systems that rely on unification.
If the unification problem has to deal with logic variables at higher
type (functional type), we speak of higher-order
unification~\cite{Huet75}.  Higher-order unification is in general
undecidable, but it can be turned decidable, if appropriately
restricted to a fragment.  For example, Miller's pattern fragment
characterizes a first-order fragment, for which unification is
decidable~\cite{Miller91jlc}. 

As substructural type theories are becoming more prevalent, for
example, in systems that need to represent consumable resources,
higher-order unification algorithms need to deal with logic variables
at linear or affine type.  Linear and affine type theories, for
example, refine intuitionistic type theory in the following way:
Besides intuitionistic assumptions, which can be referred to an
arbitrary number of times, linear and affine assumptions are treated
as resources that must be referred to \emph{exactly once} and \emph{at
  most once}, respectively.

As substructural type theories are mere refinements, one might erroneously suspect that the standard
intuitionistic pattern unification algorithm can be applied to this setting directly.
This, unfortunately, is not the case.  Consider the following two linear unification problems, where we write, as usual,~$\lhat{}$ for linear application and
juxtaposition for intuitionistic application.
\begin{align}
F\lhat{}x &\doteq c\lhat{}(H_1\;x)\lhat{}(H_2\;x) \label{eq:split2} \\
F\lhat{}x &\doteq c\lhat{}(H\;x) \label{eq:split1}
\end{align}
These examples take place in a context in which $x$ is an intuitionistic
variable.  However, the linear application on the left-hand side implies
that the variable must occur \emph{exactly} once in any valid
instantiation of $F$, but in~\eqref{eq:split2} we cannot know whether
$x$ should occur in $H_1$ or $H_2$.  This additional problem over normal
intuitionistic higher-order unification is caused exactly by the
interaction of linear and intuitionistic variables.
We solve this issue by imposing a separation of linear, affine, and
intuitionistic variables.


In this paper, we refine the intuitionistic pattern fragment into a
pattern fragment for linear and affine type theory.
We describe a unification algorithm for this fragment 
and prove and prove it correct. Furthermore, we show that in this fragment most general unifiers
exist. Finally, we extend the algorithm with a procedure we call
\emph{linearity pruning}.  This procedure goes beyond the pattern
fragment and treats equations such as~\eqref{eq:split2}
and~\eqref{eq:split1} where variables may have to change their status,
for example from being affine to linear.  Unification problems in this
extended fragment continue to be decidable.  For example,
for~\eqref{eq:split1} the algorithm finds the most general unifier, which is $F=
\widehat{\lambda} x.c \lhat{} (G\lhat{}x)$ and $H= \lambda x.G\lhat{}x$.
Our focus in this paper is finding unique most general unifiers, and
since~\eqref{eq:split2} has a set of most general unifiers of size two,
we are not going to try to solve it.  However, one could easily extend
linearity pruning to these cases by considering the finite number of context
splits.

Previous approaches to higher-order linear unification have been
restricted to highly non-deterministic algorithms, such as the
preunification by Cervesato and
Pfenning~\cite{Cervesato97linearhigher-order}.  In contrast, our
algorithm is completely deterministic, and very well suited for
implementation.  It is the core algorithm of the Celf proof
assistant~\cite{SchackNielsen08ijcar}.

\section{Language}
In~\cite{SchackNielsen10ijcar} we introduced a calculus of explicit substitutions
for the $\lambda$-calculus with linear, affine, and intuitionistic
variables and logic variables.  Along with the calculus we introduced a
type system and a reduction semantics, which was proven to be
type-preserving, confluent, and terminating.
\begin{align*}
&\textbf{Types:} & A,B &::= a \mid A\aand B \mid
  A\multimap B \mid A\affarr B \mid A\rightarrow B \\
&\textbf{Terms:} & M,N &::= 1^f \mid M[s] \mid
  \langle M,N\rangle \mid \textsf{fst}\;M \mid \textsf{snd}\;M \mid X[s] \\
  &&&\phantom{::=} \mid \widehat{\lambda}M \mid \afflam M \mid \lambda M
  \mid M\lhat{}N \mid M\affapp N \mid M\;N \\
&\textbf{Substitutions:} & s,t &::= \textsf{id} \mid \;\uparrow\; \mid
  M^f.s \mid s \circ t \\
&\textbf{Linearity flags:} & f &::= \fI \mid \fA \mid \fL \\
&\textbf{Contexts:} & \Gamma &::= \cdot \mid \Gamma,A^l \\
&\textbf{Context linearity flags:} & l &::= f \mid \fUL \mid \fUA
\end{align*}
We tag each variable $1^f$ with a flag signifying whether
the variable is intuitionistic, affine, or linear.
We 
use $\uparrow^n$ where $n\geq 0$ as a short-hand for $n$ compositions of shift,
i.e.\ $\uparrow\circ\; (\uparrow\circ\; (\ldots \;\circ\; (\uparrow\circ\uparrow)\ldots))$,
where $\uparrow^0$ means $\textsf{id}$.  Additionally, de Bruijn indices
$n^f$ with $n>1$ are short-hand for $1^f[\uparrow^{n-1}]$.
The context linearity flags and the corresponding assumptions in
contexts are denoted \emph{intuitionistic} ($\fI$),
\emph{affine} ($\fA$), \emph{used affine} ($\fUA$),
\emph{linear} ($\fL$), and \emph{used linear} ($\fUL$).

\begin{figure}[t]
\[
\ianc{}{\cdot = \cdot\Join\cdot}{}
\qquad
\ianc{\Gamma = \Gamma_1\Join\Gamma_2}{\Gamma,A^\fI = \Gamma_1,A^\fI\Join\Gamma_2,A^\fI}{}
\]
\[
\ianc{\Gamma = \Gamma_1\Join\Gamma_2}{\Gamma,A^\fUL = \Gamma_1,A^\fUL\Join\Gamma_2,A^\fUL}{}
\qquad
\ianc{\Gamma = \Gamma_1\Join\Gamma_2}{\Gamma,A^\fL = \Gamma_1,A^\fL\Join\Gamma_2,A^\fUL}{}
\qquad
\ianc{\Gamma = \Gamma_1\Join\Gamma_2}{\Gamma,A^\fL = \Gamma_1,A^\fUL\Join\Gamma_2,A^\fL}{}
\]
\[
\ianc{\Gamma = \Gamma_1\Join\Gamma_2}{\Gamma,A^\fUA = \Gamma_1,A^\fUA\Join\Gamma_2,A^\fUA}{}
\qquad
\ianc{\Gamma = \Gamma_1\Join\Gamma_2}{\Gamma,A^\fA = \Gamma_1,A^\fA\Join\Gamma_2,A^\fUA}{}
\qquad
\ianc{\Gamma = \Gamma_1\Join\Gamma_2}{\Gamma,A^\fA = \Gamma_1,A^\fUA\Join\Gamma_2,A^\fA}{}
\]
\caption{Context splitting\label{fig:ctx-join}}
\end{figure} 

\begin{figure}[t]
\[
\ianc{\textsf{nolin}(\Gamma)}{\Gamma,A^f\vdash 1^f\Rightarrow A}{}
\qquad
\ianc{\Gamma\vdash n^f\Rightarrow B\quad l\in\{\fI,\fA,\fUL,\fUA\}}{\Gamma,A^l\vdash (n+1)^f\Rightarrow B}{}
\qquad
\ianc{\Gamma\vdash M \Rightarrow a}{\Gamma\vdash
M\Leftarrow a}{}
\qquad
\ianc{\Gamma\vdash s:\Gamma_X}{\Gamma\vdash X[s] \Rightarrow  A_X}{}
\]
\[
\ianc{\Gamma\vdash M\Leftarrow A\quad \Gamma\vdash N\Leftarrow B}{\Gamma\vdash\langle
M,N\rangle \Leftarrow  A\aand B}{}
\qquad
\ianc{\Gamma\vdash M\Rightarrow A\aand B}{\Gamma\vdash\textsf{fst}\;M \Rightarrow A}{}
\qquad
\ianc{\Gamma\vdash M\Rightarrow A\aand B}{\Gamma\vdash\textsf{snd}\;M \Rightarrow B}{}
\]
\[
\ianc{\Gamma,A^\fL\vdash M \Leftarrow B}{\Gamma\vdash\widehat{\lambda}M
\Leftarrow 
A\multimap B}{}
\qquad
\ianc{\Gamma=\Gamma_1\Join\Gamma_2\quad \Gamma_1\vdash M\Rightarrow A\multimap B
\quad \Gamma_2\vdash N \Leftarrow  A}{\Gamma\vdash M\lhat{}N \Rightarrow  B}{}
\]
\[
\ianc{\Gamma,A^\fA\vdash M \Leftarrow B}{\Gamma\vdash\afflam M
\Leftarrow 
A\affarr B}{}
\qquad
\ianc{\Gamma=\Gamma_1\Join\Gamma_2\quad\textsf{nolin}(\Gamma_2)\quad
\Gamma_1\vdash M\Rightarrow A\affarr B
\quad \Gamma_2\vdash N \Leftarrow  A}{\Gamma\vdash M\affapp N \Rightarrow  B}{}
\]
\[
\ianc{\Gamma,A^\fI\vdash M \Leftarrow B}{\Gamma\vdash\lambda M \Leftarrow  A\rightarrow B}{}
\qquad
\ianc{\Gamma\vdash M\Rightarrow A\rightarrow B\quad \overline{\Gamma}\vdash N \Leftarrow  A}{\Gamma\vdash M\;N \Rightarrow  B}{}
\]
\caption{Bidirectional typing of terms in canonical form\label{fig:typ-term}}
\end{figure} 

\begin{figure}[t]
\[
\ianc{}{\cdot\vdash\;\uparrow^0\;:\cdot}{}
\qquad
\ianc{\Gamma\vdash\;\uparrow^n\; : \Gamma'\quad
l\in\{\fI,\fA,\fUL,\fUA\}}{\Gamma,A^l\vdash\;\uparrow^{n+1}\;:\Gamma'}{}
\qquad
\ianc{\overline{\Gamma}\vdash M\Leftarrow A\quad
\Gamma\vdash s:\Gamma'}{\Gamma\vdash M^\fI.s :
\Gamma',A^\fI}{}
\]
\[
\ianc{\Gamma = \Gamma_1\Join\Gamma_2\quad
\Gamma_1\vdash M\Leftarrow A\quad
\Gamma_2\vdash s:\Gamma'}{\Gamma\vdash M^\fL.s :
\Gamma',A^\fL}{}
\qquad
\ianc{\underline{\Gamma}\vdash_i M \Leftarrow A\quad \Gamma\vdash
s:\Gamma'}{\Gamma\vdash M^\fL.s :
\Gamma',A^\fUL}{}
\]
\[
\ianc{\Gamma = \Gamma_1\Join\Gamma_2\quad\textsf{nolin}(\Gamma_1)\quad
\Gamma_1\vdash M\Leftarrow A\quad
\Gamma_2\vdash s:\Gamma'}{\Gamma\vdash M^\fA.s :
\Gamma',A^\fA}{}
\qquad
\ianc{\underline{\Gamma}\vdash_i M \Leftarrow A\quad \Gamma\vdash
s:\Gamma'}{\Gamma\vdash M^\fA.s :
\Gamma',A^\fUA}{}
\]
\caption{Typing of substitutions\label{fig:typ-sub}}
\end{figure} 

In this paper we will work exclusively with the corresponding calculus of canonical
forms and hereditary substitutions.  This can be obtained simply by
viewing each term as a short-hand for its unique normal form and
assuming that everything is fully $\eta$-expanded.  The resulting type
system is shown in Figures~\ref{fig:ctx-join}--\ref{fig:typ-sub}.
We write $\Gamma\vdash M:A$ as a shorthand for either $\Gamma\vdash M\Leftarrow A$
or $\Gamma\vdash M\Rightarrow A$.

The intuitionistic part of a context $\overline{\Gamma}$ is formed by
rendering all linear and affine variables unavailable, which corresponds to updating
the context linearity flags from
$\fL$ to $\fUL$ and $\fA$ to $\fUA$.
Similarly, the largest context that can split to a given context is
denoted $\underline{\Gamma}$ and constructed by changing every
$\fUL$ to $\fL$ and $\fUA$ to $\fA$.
The predicate $\textsf{nolin}(\Gamma)$ specifies that no linear
assumptions occur in $\Gamma$, i.e.\ no flag in $\Gamma$ is equal to
$\fL$.  The relaxed typing judgment $\Gamma\vdash_i M:A$ is
similar to $\Gamma\vdash M:A$ except that it makes all variables
available everywhere disregarding linearity and affineness.
The typing judgments could be augmented with an additional kind of
context for looking up logic variables, but we will keep this lookup
implicit and simply write $\Gamma_X$ and $A_X$ for the context and type
of a logic variable $X$.

\begin{figure}
\begin{align*}
(\widehat{\lambda}M)\lhat{}N &= M[N^\fL.\textsf{id}] & (\widehat{\lambda}M)[s] &= \widehat{\lambda}(M[1^{\fL\fL}.(s\;\circ\uparrow)]) & 1^f[M^{f}.s] &= M \\
(\afflam M)\affapp N &= M[N^\fA.\textsf{id}]         & (\afflam M)[s] &= \afflam (M[1^{\fA\fA}.(s\;\circ\uparrow)])                   & M[\textsf{id}] &= M \\
(\lambda M)\;N &= M[N^\fI.\textsf{id}]               & (\lambda M)[s] &= \lambda (M[1^{\fI\fI}.(s\;\circ\uparrow)])                   & \textsf{id}\circ s &= s \\
\textsf{fst}\langle M,N\rangle &= M                  & (M\lhat{}N)[s] &= M[s]\lhat{}N[s]                                              & s\circ\textsf{id} &= s \\
\textsf{snd}\langle M,N\rangle &= N                  & (M\affapp N)[s] &= M[s]\affapp N[s]                                            & \uparrow\circ\; (M^f.s) &= s \\
\langle M,N\rangle[s] &= \langle M[s],N[s]\rangle    & (M\;N)[s] &= M[s]\;N[s]                                                        & (M^f.s)\circ t &= M[t]^f.(s\circ t) \\
(\textsf{fst}\;M)[s] &= \textsf{fst}\;(M[s])         & M[s][t] &= M[s\circ t]                                                         & (s_1\circ s_2)\circ s_3 &= s_1\circ (s_2\circ s_3) \\
(\textsf{snd}\;M)[s] &= \textsf{snd}\;(M[s])         & X[s][t] &= X[s\circ t]                                                         & \uparrow^n &= (n+1)^{ff}.\uparrow^{n+1}
\end{align*}
\caption{Equalities\label{fig:equalities}}
\end{figure} 

Restricting ourselves to canonical forms while retaining the syntax of
redices and closures as short-hands for their corresponding normal forms
induces equalities corresponding to the rewrite rules of the original
system.  The induced equalities are shown in Figure~\ref{fig:equalities}.
Additionally, the two typing rules for $M[s]$ and $s_1\circ s_2$ from~\cite{SchackNielsen10ijcar},
which are left out, are now simply admissible rules proving type
preservation of hereditary substitution:
\[
\ianc{\Gamma\vdash s:\Gamma'\quad \Gamma'\vdash M :A}{\Gamma\vdash
M[s]:A}{}
\qquad
\ianc{\Gamma\vdash s_2:\Gamma''\quad\Gamma''\vdash
s_1:\Gamma'}{\Gamma\vdash s_1\circ s_2:\Gamma'}{}
\]

We use \emph{spine notation}~\cite{Cervesato03jlc} as a convenient short-hand for series of
applications and projections:
\[
S::= () \mid M;S \mid M\affext S \mid M\lhat{;}S \mid \textsf{fst};S \mid
\textsf{snd};S
\]
The term $M\cdot S$ is short-hand for the term where all the terms and
projections in $S$ are applied to $M$ as follows:
\begin{align*}
M\cdot () &= M & M\cdot (N\affext S) &= (M\affapp N)\cdot S &
  M\cdot (\textsf{fst};S) &= (\textsf{fst}\;M)\cdot S \\
M\cdot (N;S) &= (M\;N)\cdot S & M\cdot (N\lhat{;}S) &= (M\lhat{}N)\cdot S &
  M\cdot (\textsf{snd};S) &= (\textsf{snd}\;M)\cdot S
\end{align*}
We write $S[s]$ for the argumentwise application of $s$ in $S$ and observe
that $(M\cdot S)[s] = M[s]\cdot S[s]$.

We write $[X\leftarrow N]M$ for the \emph{instantiation of the logic variable} $X$ with term $N$ in
term $M$.  This instantiation is type preserving, which follows 
by induction on $M$ and the subject reduction property of hereditary substitutions.
\begin{thm}\label{thm:modalcut}
If $\Gamma_X\vdash N: A_X$ and $\Gamma\vdash M:A$ then
$\Gamma\vdash [X\leftarrow N]M:A$.
\end{thm}
Theorem~\ref{thm:modalcut} is also called the contextual modal cut
admissibility theorem for linear and affine contextual modal logic.

\section{Patterns}\label{sec:inversion}
The hallmark characteristic of the intuitionistic pattern fragment is
the invertibility of substitutions~\cite{Dowek98tr}.  Our pattern
fragment for the linear and affine calculus that we are going to
introduce next continues to guarantee this important property.

Consider a substitution
$\Gamma\vdash a_1^{f_1}\ldots a_p^{f_p}.\uparrow^n\;:\Gamma'$.
Assume that $a_j$ is a variable $n_j^{f'_j}$.  We say the
substitution extension
$n_j^{f'_j f_j}$ is
\emph{linear} if
$f'_j f_j={\fL\fL}$, it is \emph{affine} if 
$f'_j f_j={\fA\fA}$, it is \emph{intuitionistic} if 
$f'_j f_j={\fI\fI}$, and it is \emph{linear-changing} if 
$f'_j f_j={\fI\fL}$, $f'_j f_j={\fI\fA}$, or $f'_j f_j={\fA\fL}$.
Notice that the possibilities $\fL\fI$, $\fA\fI$, and $\fL\fA$
cannot occur in well-typed substitutions since this would imply
referencing a linear or affine assumption in an intuitionistic context or
a linear assumption in an affine context.

\begin{defin} \label{def:patsub}
A substitution $\Gamma\vdash a_1^{f_1} \ldots a_p^{f_p}.\uparrow^n\;:\Gamma'$
is said to be a \emph{pattern substitution} if all the terms $a_j$ for
$j\in \{1,\ldots,p\}$ are distinct de Bruijn indices and none of them
are linear-changing extensions in the substitution.
A pattern substitution is called a \emph{weakening substitution} if the
indices $a_j$ form an increasing sequence.
\end{defin}

Note that in a pattern substitution all de Bruijn indices are less than or equal to $n$ since $n$ is
equal to the length of $\Gamma$.   To understand pattern
substitutions in the presence of logic variables during
lowering (discussed in Section~\ref{sec:unification}), we define the
\emph{extension of pattern substitution $s$ by spine $S$}, written as $S.s$:
\begin{align*}
().s &= s &              (N\affext S).s &= S.(N^\fA.s) & (\textsf{fst};S).s &= S.s \\
(N;S).s &= S.(N^\fI.s) & (N\lhat{;}S).s &= S.(N^\fL.s) & (\textsf{snd};S).s &= S.s
\end{align*}


\begin{defin}
A term $M$ is said to be a \emph{pattern} or within the
\emph{pattern fragment} if all
occurrences of logic variables $X[s]\cdot S$ satisfy the property that
the substitution $S.s$ is a pattern substitution.
\end{defin}

Recall example~\eqref{eq:split2} from the introduction.  In our system,
the equation is written as
$F[\uparrow^1]\cdot(1^{\fI}\lhat{;}())\doteq
c\cdot(H_1[\uparrow^1]\cdot(1^{\fI};())\lhat{;}H_2[\uparrow^1]\cdot(1^{\fI};())\lhat{;}())$.
We observe that it is not a pattern since there is a linear-changing substitution
extension on the left-hand side in $(1^{\fI}\lhat{;}()).\uparrow^1=1^{\fI\fL}.\uparrow^1$.

It can be proven that the pattern fragment is stable under hereditary
substitution, logic variable instantiation, and inversion of
substitutions.  In particular, the following two theorems hold:

\begin{thm}
The pattern fragment is stable under logic variable instantiation.
I.e.\ for any patterns $M$ and $N$, $[X\leftarrow N]M$ is a pattern.
\end{thm}

\begin{thm}
If $s$ is a pattern substitution and $M[s]$ is a
pattern then $M$ is a pattern.
\end{thm}

The proofs are relatively straight-forward extensions of the proofs given
in~\cite{Dowek98tr} for the intuitionistic pattern fragment.

Next, we define the inverse of a pattern substitution. The name is
justified by Theorem~\ref{thm:inverse} below.
\begin{defin}
Let $s=a_1^{f_1}\ldots a_p^{f_p}.\uparrow^n$ be a pattern substitution.
We define its inverse to be
$s^{-1}=e_1^{g_1}\ldots e_n^{g_n}.\uparrow^p$ where
$e_j^{g_j}=i^{f_i f_i}$
when $a_i=j^{f_i}$ and $e_j$ is undefined otherwise.
The undefined extensions $e_j^{g_j}$ are flagged  intuitionistic,
affine, or linear depending on the $j$th assumption in the codomain
of $s$.
\end{defin}
Intuitively, this definition is well defined: the $a_i$s are
distinct and less than or equal to $n$.  For the undefined $e_j$ one can think
of an arbitrary term of the right type, e.g.\ a freshly created logic variable.


In the following we will refer to affine weakening on contexts $\Gamma\affweak\Gamma'$, which is defined as
\[
\Gamma\affweak\Gamma' \;\equiv\;
\exists \Gamma''.\;\;\Gamma=\Gamma''\Join\Gamma'\;\land\;\textsf{nolin}(\Gamma'')
\]
Notice that affine weakening is reflexive and transitive, as it merely
amounts to changing some number of $\fA$s into $\fUA$s.

\begin{lem}
For a pattern substitution $\Gamma_2\vdash s:\Gamma'$ there exists a
$\Gamma_1$ with $\Gamma_2\affweak\Gamma_1$ such that
$\Gamma_1\vdash s:\Gamma'$ and the inverse is well-typed with
$\Gamma'\vdash s^{-1}:\Gamma_1$.
\end{lem}
\begin{proof}
Let $s=a_1^{f_1}\ldots a_p^{f_p}.\uparrow^n$.
Then $\Gamma_2=\cdot,B_n^{l^2_n},\ldots,B_1^{l^2_1}$
and $\Gamma'=\cdot,A_p^{l'_p},\ldots,A_1^{l'_1}$.  Intuitively we are
going to take $\Gamma_1$ to be the smallest possible such that $s$ is
still well-typed, i.e.\ we are going to make all the affine assumptions
that are not used in $s$ unavailable.  More formally we are going to set
$\Gamma_1=\cdot,B_n^{l^1_n},\ldots,B_1^{l^1_1}$ where $l^1_j=l^2_j$
when $l^2_j\in\{\fI,\fL,\fUL,\fUA\}$.  When $l^2_j=\fA$ the $l^1_j$ will
be defined below.

Consider each variable $a_i^{f_i}=j^{f_i f_i}$ in $s$.
Note that we have $A_i=B_j$.
If $l'_i=f$ where $f$ is either $\fI$ or $\fL$ then we have
$f_i=f$ and $l^2_j=l^1_j=f$.
In the case where $l'_i=\fUL$ then $f_i=\fL$ and $l^2_j=l^1_j$
are either equal to $\fUL$ or $\fL$, but since all the variables in $s$
are distinct it has to be $\fUL$.
If $l'_i=\fA$ then $f_i=\fA$ and $l^2_j=\fA$, and in this case we set
$l^1_j=\fA$.
Finally, if $l'_j=\fUA$ then $f_i=\fA$ and $l^2_j$ is either $\fUA$ or
$\fA$.  If $l^2_j=\fUA$ then $l^1_j$ is also equal to $\fUA$, and if
$l^2_j=\fA$ then we can set $l^1_j=\fUA$ since $j$ does not occur
anywhere else in $s$.  This means that for all defined extensions
$e_j=i^{f_i}$ in $s^{-1}$ we have $l'_i=l^1_j$.

The remaining $B_j^{l^2_j}$s for which there are no $a_i=j^{f_i}$ are
all shifted away by the $\uparrow^n$ part of $s$.  Therefore none of them
can be linear, and if any of them are affine, i.e.\ have $l^2_j=\fA$,
we set $l^1_j=\fUA$.
This means that all the undefined extensions in $s^{-1}$ correspond to
intuitionistic, used linear, or used affine assumptions in $\Gamma_1$,
and we see that 
$s^{-1}$ indeed is well-typed with $\Gamma'\vdash s^{-1}:\Gamma_1$.
\end{proof}
\begin{thm}\label{thm:inverse}
Given a pattern substitution $\Gamma\vdash s:\Gamma'$, we have
$\Gamma'\vdash s\circ s^{-1}:\Gamma'$ and
$s\circ s^{-1}=\textup{\textsf{id}}$.
\end{thm}
\begin{proof}
Let $s=a_1^{f_1}\ldots a_p^{f_p}.\uparrow^n$.
Since $a_i=j^{f_i}$ then the $j$th extension in $s^{-1}$ is equal to
$i^{f_i}$, and thus $a_i[s^{-1}] = i^{f_i}$ for all $i$.
\end{proof}

We have the usual definition of occurrence, rigid occurrence, and flexible
occurrence written as $\in$, $\in_{\textsf{rig}}$, and $\in_{\textsf{flex}}$
respectively.  These relations are only defined for canonical forms in
which all logic variables are of base type (lowering will achieve this).
Occurrence is defined as
$\in\;=\;\in_{\textsf{rig}}\cup\in_{\textsf{flex}}$.  Rigid and flexible
occurrence are defined as follows, where we write $\in_*$ for either
rigid or flexible occurrence.
\[
\ianc{}{n\in_{\textsf{rig}} n^f}{}
\quad
\ianc{n\in_{\textsf{rig}} s}{n\in_{\textsf{flex}} X[s]}{}
\quad
\ianc{a_i=n^f}{n\in_{\textsf{rig}} a_1^{f_1}\ldots a_p^{f_p}.\uparrow^m}{}
\quad
\ianc{n\in_* M_i}{n\in_* \langle M_1,M_2\rangle}{}
\quad
\ianc{n\in_* M}{n\in_* \textsf{fst}\; M}{}
\quad
\ianc{n\in_* M}{n\in_* \textsf{snd}\; M}{}
\]
\[
\ianc{n+1\in_* M}{n\in_* \widehat{\lambda}M}{}
\quad
\ianc{n+1\in_* M}{n\in_* \afflam M}{}
\quad
\ianc{n+1\in_* M}{n\in_* \lambda M}{}
\quad
\ianc{n\in_* M_i}{n\in_* M_1\lhat{}M_2}{}
\quad
\ianc{n\in_* M_i}{n\in_* M_1\affapp M_2}{}
\quad
\ianc{n\in_* M_i}{n\in_* M_1\;M_2}{}
\]
If $n\in_{\textsf{flex}} M$ then the definition implies that there is
some logic variable \mbox{$X[a_1^{f_1}\ldots a_p^{f_p}.\uparrow^m]$} in
$M$ beneath $k$ lambdas such that
$(n+k)^{f_i}=a_i$.  In this case we say that
$n$ occurs in the $i$th argument of $X$.

\begin{lem}\label{lem:linoccur}
Linearity implies occurrence.
\begin{enumerate}
\item
Let $\Gamma\vdash s:\Gamma'$ be a pattern substitution and the $n$th
assumption in $\Gamma$ be linear.  Then $n$ occurs in $s$.
\item
Let $\Gamma\vdash M:A$ be a pattern and let the $n$th
assumption in $\Gamma$ be linear.  Then $n$ occurs in $M$.
\end{enumerate}
\end{lem}
\begin{proof}
If $s=a_1^{f_1f_1}\ldots a_p^{f_pf_p}.\uparrow^m$ then we must have
$n=a_i$ for some $a_i$ since a linear assumption cannot be shifted away.
The second case is by induction on $M$.
\end{proof}

\begin{defin}
Given the typing of a substitution $\Gamma\vdash s:\Gamma'$ we will
call it \emph{strong} if there exists no $\Gamma''\neq\Gamma'$ such
that $\Gamma''\affweak\Gamma'$ and $\Gamma\vdash s:\Gamma''$.
\end{defin}
For a pattern substitution
$\cdot,B_n^{l_n},\ldots,B_1^{l_1}\vdash
a_1^{f_1}\ldots a_p^{f_p}.\uparrow^n :
\cdot,A_p^{l'_p},\ldots,A_1^{l'_1}$
we see that it is strong if and only if
for each affine variable $a_i=j^\fA$ we have $l'_i=\fUA$ implies
$l_j=\fUA$.

Consider the split of a strong pattern substitution
$\Gamma\vdash s:\Gamma'$ over a context split
$\Gamma'=\Gamma'_1\Join\Gamma'_2$ into $\Gamma_1\vdash s:\Gamma'_1$ and
$\Gamma_2\vdash s:\Gamma'_2$ with $\Gamma=\Gamma_1\Join\Gamma_2$.
For any used affine assumption in $\Gamma'_1$ the assumption is either
affine or used affine in $\Gamma'$ and $\Gamma'_2$.  If it is used affine then the
corresponding assumption is also used affine in $\Gamma$ and thereby
$\Gamma_1$.  If it is affine then the corresponding assumption has to be
affine in $\Gamma_2$ and is thereby used affine in $\Gamma_1$.  This
means that $\Gamma_1\vdash s:\Gamma'_1$ is strong and by symmetry so is
$\Gamma_2\vdash s:\Gamma'_2$.

\begin{thm}\label{thm:occurrence}
Let $\Gamma\vdash s:\Gamma'$ be a pattern substitution and
$\Gamma\vdash M:A$ be a term in which all logic variables
are of base type.
\begin{enumerate}
\item
If there exists a term $\Gamma'\vdash M':A$ such that $M=M'[s]$
then every variable occurring in $M$ also occurs in $s$.
\item
If the typing $\Gamma\vdash s:\Gamma'$ is strong and every variable
occurring in $M$ also occurs in $s$ then there exists a term
$\Gamma'\vdash M':A$ such that $M=M'[s]$.
\end{enumerate}
\end{thm}
\begin{proof}
1.\ follows by induction on $M'$ and 2.\ by induction on $M$  using the
fact that context splits preserve a strong typing of $s$.  It is easy to see
that a strong typing of $s$ implies a strong typing of
$1^{ff}.(s\;\circ\uparrow)$ when going beneath a lambda-binder.

For the base case
$M=n^f$ we get that $n\in s$ implies that the $n$th assumption in
$\Gamma$ corresponds to an assumption, say the $m$th, in $\Gamma'$.
Now, we can take $M'=m^f$, and since $s$ is strong, availability of the
$n$th assumption in $\Gamma$ implies availability of the $m$th
assumption in $\Gamma'$ and thus that $M'$ is well-typed.
The base case $M=X[t]$ is similar, when noting that the shift at the end
of $s$ is equal to the shift at the end of $t$, since they
are both equal to the length of $\Gamma$.
\end{proof}

Theorem~\ref{thm:occurrence} states that occurrence is a
conservative approximation of the set of
variables occurring in any instantiation of a term,
i.e.\ if $n\in [X\leftarrow N]M$ then $n\in M$.  The opposite is
not necessarily true.

\section{Pattern unification}
A unification problem $P$ is a conjunction of unification
equations, and a solution to a unification problem is an instantiation
of the logic variables such that all equations are satisfied.  Such a
collection of logic variable instantiations will be written as $\theta$
and we say that $\theta$ solves $P$.  In this section we describe an
algorithm that returns ``no'' if no such solution exists or a
most general unifier otherwise, i.e.\ a solution that all other solutions are
refinements of.

More formally, we write $\Gamma\vdash M_1\doteq M_2 :A$ for a unification equation 
or simply $M_1\doteq M_2$ with the implicit understanding that both
terms have the same type in the same context.
Unification equations are symmetric and we will implicitly
switch from $M_1\doteq M_2$ to $M_2\doteq M_1$ when needed.  Unification
problems are given by the following grammar, where $\mathbb{T}$ is the
solved unification problem and $\mathbb{F}$ is the unification problem
with no solutions.
\[
P::= \mathbb{T} \mid \mathbb{F} \mid P\land (\Gamma\vdash M_1\doteq M_2 :A)
\]
For convenience we generalize unification equations to spines
and write $S_1 \doteq S_2 $ as a
short-hand for the argumentwise conjunction of unification equations
(see below).

\subsection{Unification algorithm}\label{sec:unification}
The unification algorithm consists of a set of
transformation rules of the form $P\mapsto P'$.   We will see 
that the repeated application of these rule to any unification problem 
will eventually terminate resulting in  either $\mathbb{F}$, which indicates that the original problem has
no solution, or $\mathbb{T}$, which indicates that all equations have been
solved and that a most general unifier has been found.  In this case the
most general unifier is a mapping from logic variables to their
instantiations as computed during the execution of the
algorithm.  The unification algorithm is given in
Figure~\ref{fig:pat-unif} and each rule is explained in detail below.
For convenience we write the decomposition of a term $M$ into one of
its subterms $N$ and the surrounding term with a hole $M'\{\cdot\}$ as $M=M'\{N\}$.


\begin{figure}
\[
{\newcommand{\spacer}{\rule{0pt}{3.1ex}}
\newcommand{\width}{.75\textwidth}
\begin{array}{llll}
\textbf{dec-lam-l} &
P\land\widehat{\lambda} M_1 \doteq \widehat{\lambda} M_2
&\mapsto& P\land M_1\doteq M_2 \\
\spacer\textbf{dec-lam-a} &
P\land\afflam M_1 \doteq \afflam M_2
&\mapsto& P\land M_1\doteq M_2 \\
\spacer\textbf{dec-lam-i} &
P\land\lambda M_1 \doteq \lambda M_2
&\mapsto& P\land M_1\doteq M_2 \\
\spacer\textbf{dec-pair} &
P\land\langle M_1,N_1\rangle \doteq \langle M_2,N_2\rangle
&\mapsto& P\land M_1\doteq M_2\land N_1\doteq N_2 \\
\spacer\textbf{dec-atomic-eq} &
P\land n^f\cdot S_1 \doteq n^f\cdot S_2
&\mapsto& P\land S_1\doteq S_2 \\
\spacer\textbf{dec-atomic-neq} &
P\land n^f\cdot S_1 \doteq m^{f'}\cdot S_2
&\mapsto& \mathbb{F} \\
&\multicolumn{3}{l}{
\textrm{if }n\neq m
} \\
\spacer\textbf{lower-lolli} &
P
&\mapsto& [X\leftarrow \widehat{\lambda} Y[\textsf{id}]]P \\
&\multicolumn{3}{p{\width}}{
if $A_X=A\multimap B$ and $Y$ is a fresh logic variable
with $A_Y=B$ and $\Gamma_Y=\Gamma_X,A^\fL$
} \\
\spacer\textbf{lower-affarr} &
P
&\mapsto& [X\leftarrow \afflam Y[\textsf{id}]]P \\
&\multicolumn{3}{p{\width}}{
if $A_X=A\affarr B$ and $Y$ is a fresh logic variable
with $A_Y=B$ and $\Gamma_Y=\Gamma_X,A^\fA$
} \\
\spacer\textbf{lower-arr} &
P
&\mapsto& [X\leftarrow \lambda Y[\textsf{id}]]P \\
&\multicolumn{3}{p{\width}}{
if $A_X=A\rightarrow B$ and $Y$ is a fresh logic variable
with $A_Y=B$ and $\Gamma_Y=\Gamma_X,A^\fI$
} \\
\spacer\textbf{lower-and} &
P
&\mapsto& [X\leftarrow \langle Y[\textsf{id}],Z[\textsf{id}]\rangle]P \\
&\multicolumn{3}{p{\width}}{
if $A_X=A\aand B$ and $Y$ and $Z$ are fresh logic variables
with $A_{Y}=A$, $A_{Z}=B$, $\Gamma_{Y}=\Gamma_X$, and $\Gamma_{Z}=\Gamma_X$
} \\
\spacer\textbf{occurs-check} &
P\land X[s] \doteq n^f\cdot S\{X[t]\}
&\mapsto& \mathbb{F} \\
\spacer\textbf{pruning-fail} &
P\land X[s] \doteq M
&\mapsto& \mathbb{F} \\
&\multicolumn{3}{p{\width}}{
if $n\notin s$ and $n\in_{\textsf{rig}} M$
} \\
\spacer\textbf{pruning} &
P\land X[s] \doteq M
&\mapsto& [Y\leftarrow Z[w]](P\land X[s]\doteq M) \\
&\multicolumn{3}{p{\width}}{
if $n\notin s$, $n$ occurs flexibly in $M$ in the $i$th argument of
the logic variable $Y$, $w=\textsf{weaken}(\Gamma_Y;i)$,
and $Z$ is a fresh logic variable with
$A_Z=A_Y$ and $\Gamma_Z=\Gamma_Y\div i$
} \\
\spacer\textbf{ctx-pruning} &
P
&\mapsto& [X\leftarrow Y[w]]P \\
&\multicolumn{3}{p{\width}}{
if $\Gamma_X=\cdot,A_p^{l_p},\ldots,A_1^{l_1}$ with
$l_n\in\{\fUA,\fUL\}$,
$w=\textsf{weaken}(\Gamma_X;n)$, and $Y$ is a fresh logic
variable with $A_Y=A_X$ and
$\Gamma_Y=\Gamma_X\div n$
} \\
\spacer\textbf{instantiation} &
P\land X[s] \doteq M
&\mapsto& [X\leftarrow M[s^{-1}]]P \\
&\multicolumn{3}{p{\width}}{
if $X$ does not occur in $M$,
$\Gamma_X$ contains no used affine assumptions,
and $n\in M$ implies $n\in s$
} \\
\spacer\textbf{intersection-eq} &
P\land X[s] \doteq X[s]
&\mapsto& P \\
\spacer\textbf{intersection-fail} &
P\land X[s] \doteq X[t]
&\mapsto& \mathbb{F} \\
&\multicolumn{3}{p{\width}}{
if $s\neq t$ and $s\cap t$ does not exist
} \\
\spacer\textbf{intersection} &
P\land X[s] \doteq X[t]
&\mapsto& [X\leftarrow Y[s\cap t]]P \\
&\multicolumn{3}{p{\width}}{
if $s\neq t$, $s\cap t$ exists, and $Y$ is a fresh logic variable with
$A_Y=A_X$ and $\Gamma_Y$ equal to the domain of the weakening
substitution $s\cap t$
} \\
\end{array}
}
\]
\caption{Pattern unification rules\label{fig:pat-unif}}
\end{figure} 

\medskip \noindent\textbf{Decomposition.}
Consider a unification equation $\Gamma\vdash M_1\doteq M_2 :A$ and
assume that $A$ is not a base type.

If $A=B\multimap C$ then we
must have $M_1=\widehat{\lambda}M'_1$ and $M_2=\widehat{\lambda}M'_2$.
In this case $M_1$ is equal to $M_2$ under some $\theta$
if and only if $M'_1$ is equal to $M'_2$ under $\theta$ and we therefore
apply \textbf{dec-lam-l}.  The other non-base type cases for $A$ are
similar and give rise to \textbf{dec-lam-a}, \textbf{dec-lam-i}, and
\textbf{dec-pair}.

If $A$ is a base type then $M_1=H_1\cdot S_1$ and $M_2=H_2\cdot S_2$
where $H_1$ and $H_2$ are either variables or logic variables.  The case
of logic variables is handled below.  We therefore have
$n^f\cdot S_1\doteq m^{f'}\cdot S_2$.  If $n\neq m$ then no $\theta$
can make the two equal and we can therefore apply
\textbf{dec-atomic-neq}.  If $n=m$ then the
spines must unify and we apply \textbf{dec-atomic-eq} where
$P\land S_1\doteq S_2$ is defined as:
\[
\begin{array}{llll}
P\land ()\doteq () &= P                                                                     & P\land (\textsf{fst}; S_1)\doteq (\textsf{fst}; S_2) &= P\land S_1\doteq S_2 \\
P\land (M_1; S_1)\doteq (M_2; S_2) &= P\land M_1\doteq M_2\land S_1\doteq S_2       \qquad  & P\land (\textsf{snd}; S_1)\doteq (\textsf{snd}; S_2) &= P\land S_1\doteq S_2 \\
P\land (M_1\affext S_1)\doteq (M_2\affext S_2) &= P\land M_1\doteq M_2\land S_1\doteq S_2   & P\land (\textsf{fst}; S_1)\doteq (\textsf{snd}; S_2) &= \mathbb{F} \\
P\land (M_1\lhat{;} S_1)\doteq (M_2\lhat{;} S_2) &= P\land M_1\doteq M_2\land S_1\doteq S_2 & P\land (\textsf{snd}; S_1)\doteq (\textsf{fst}; S_2) &= \mathbb{F} \\
\end{array}
\]
No other cases can occur because $n=m$ trivially imply that
they have the same type.

\medskip \noindent\textbf{Lowering.}
When a logic variable occurs in a unification problem in the form
$X[s]\cdot S$ with a non-empty spine, we know that $A_X$ cannot be a
base type.  And since canonical forms of non-base type have unique head
constructors, we can safely instantiate $X$ to that particular
constructor.  This is accomplished by the rules \textbf{lower-*}.
Therefore we can assume that all logic variables are of base type.

\medskip \noindent\textbf{Occurs check.}
Consider a unification equation of the form $X[s]\doteq M$.  If $X$
also occurs in the right-hand side then either $M=n^f\cdot S\{X[t]\}$ or
$M=X[t]$.  The latter case is handled below in \textbf{Intersection}.
In the former case we have the equation $X[s]\doteq n^f\cdot S\{X[t]\}$.
Since a pattern substitution $t$ applied to any term can never alter the
shape of the term but only rename variables this equation has no
solutions, and we can apply \textbf{occurs-check}.

\medskip \noindent\textbf{Pruning.}
When we have $X[s]\doteq M$ then Theorem~\ref{thm:occurrence} tells us
that under some $\theta$ solving the equation,
variables that do not occur in $s$ cannot occur in $M$.  Assume
that $n\notin s$ and $n\in M$.  If $n\in_{\textsf{rig}} M$ then no
instantiation of logic variables can get rid of the occurrence and
we apply \textbf{pruning-fail}.  If on the other hand
$n\in_{\textsf{flex}} M$ then the occurrence is in the $i$th argument
of some logic variable $Y$.  This means, however, that no instantiation
of $Y$ in a solution can contain~$i$.
By Lemma~\ref{lem:linoccur} we know that
$n$ cannot refer to a linear assumption in the context in which $X[s]$
and $M$ are typed and therefore the $i$th assumption in $\Gamma_Y$
cannot be linear.\footnote{Notice that this argument relies on the fact
that $Y$ is under a pattern substitution and thus has no linear-changing
variables.}
Let $w$ be the weakening substitution $\textsf{weaken}(\Gamma_Y;i)$
where \textsf{weaken} is defined as:
\[
\begin{array}{llll}
\textsf{weaken}(\Gamma,A^l;1) &=& \uparrow & \textrm{if }l\neq\fL \\
\textsf{weaken}(\Gamma,A^\fI;i+1) &=&
1^{\fI\fI}.\textsf{weaken}(\Gamma;i)\;\circ \uparrow \\
\textsf{weaken}(\Gamma,A^l;i+1) &=&
1^{\fA\fA}.\textsf{weaken}(\Gamma;i)\;\circ \uparrow &
\textrm{if }l\in\{\fA,\fUA\} \\
\textsf{weaken}(\Gamma,A^l;i+1) &=&
1^{\fL\fL}.\textsf{weaken}(\Gamma;i)\;\circ \uparrow &
\textrm{if }l\in\{\fL,\fUL\} \\
\end{array}
\]
Define $\Gamma\div i$ to be the context $\Gamma$ with the $i$th assumption removed.
We see that $\Gamma\vdash\textsf{weaken}(\Gamma;i):\Gamma\div i$.
Furthermore, this is a strong typing.
Since the $i$th assumption in $\Gamma_Y$ is not linear then
$w=\textsf{weaken}(\Gamma_Y;i)$ does indeed exist.
Theorem~\ref{thm:occurrence} tells us that $Y$ has to be
instantiated to something on the form $M'[w]$ and we can therefore apply
\textbf{pruning}.

\medskip \noindent\textbf{Context pruning.}
If a logic variable $X$ is declared in context
$\Gamma_X=\cdot,A_p^{l_p},\ldots,A_1^{l_1}$ with $l_n\in\{\fUA,\fUL\}$,
we know that $n$ cannot occur in a well-typed instantiation of $X$.
Therefore, by Theorem~\ref{thm:occurrence}, $X$ has to be
instantiated to something on the form $M[\textsf{weaken}(\Gamma_X;n)]$
and we can therefore apply \textbf{ctx-pruning}.

Note that pruning the context of $X$ in this way in the case of
$X[s]\doteq M$ may allow further pruning in $M$.  Additionally, repeated
applications of this step will ensure that no used affine assumptions
occur in the context of logic variables.  Therefore all
typings of the associated substitutions are strong.

\medskip \noindent\textbf{Instantiation.}
Consider the unification equation $X[s]\doteq M$ where all used affine assumptions have been
pruned from $\Gamma_X$ and the typing of $s$ therefore is strong.
If all $n\in M$ also occur in $s$ then
Theorem~\ref{thm:occurrence} tells us that $M$ is equal to
$M'[s]$ for some $M'$.  By Theorem~\ref{thm:inverse} we know that
$M'$ is equal to $M[s^{-1}]$ and we can therefore
instantiate $X$ by the rule \textbf{instantiation} provided that $X$
does not occur in $M$.

\medskip \noindent\textbf{Intersection.}
The final case is when we have $X[s]\doteq X[t]$.  If $s=t$ then the
equation will be trivially satisfied no matter what term $X$ might be
instantiated to, so we can simply remove the equation by the rule
\textbf{intersection-eq}.  

Consider an instantiation of $X$ to some $M$.  If for all
$n\in M$ we have $n[s]=n[t]$ then the equation is
clearly satisfied.  If on the other hand there is some $n\in M$ such
that $n[s]\neq n[t]$ then the two sides of the equation
will not be equal.  Therefore any variable $n$ for which
$n[s]\neq n[t]$ cannot occur in an instantiation of $X$.
If such an $n$ is linear then Lemma~\ref{lem:linoccur} tells us that
$n$ has to occur in all instantiations and we can conclude that there is
no solution and apply \textbf{intersection-fail}.
Otherwise, any instantiation of $X$ has to be on the form $M'[s\cap t]$ for some
$M'$ where $s\cap t$ is defined as the following weakening substitution:
\[
\begin{array}{llll}
M^{f}.s\cap M^{f}.t &=&
1^{ff}.(s\cap t)\;\circ \uparrow \\
n^{ff}.s\cap m^{ff}.t &=&
(s\cap t)\;\circ \uparrow & \textrm{if }n\neq m\textrm{ and }f\in\{\fI,\fA\} \\
\uparrow^n\cap \uparrow^n &=& \textsf{id} \\
\end{array}
\]
Note that $s\cap t$ exists exactly when $n[s]=n[t]$
for all linear $n$.  The domain of
$s\cap t$ is seen to be $\Gamma_X$ with those assumptions removed for which
$n[s]\neq n[t]$.
This step is summarized by the rule \textbf{intersection}.

\subsection{Correctness}
Correctness of the unification algorithm has three parts: preservation,
progress, and termination.
\begin{thm}\label{thm:unif-correct}
The unification algorithm solves all pattern unification problems
correctly.
\begin{enumerate}
\item
If $P\mapsto P'$ then the set of solutions to $P$ is equal to the set of
solutions to $P'$.
\item
If $P$ has unsolved equations (i.e.\ $P$ is not equal to $\mathbb{F}$ or
$\mathbb{T}$) then there exists a $P'$ such that $P\mapsto P'$.
\item
The unification algorithm terminates.
\end{enumerate}
\end{thm}
\begin{proof}
The discussion above in section~\ref{sec:unification}
proves preservation of solutions~(1) and progress~(2).
For termination~(3) we will consider the lexicographic ordering of
\begin{enumerate}
\item
The total size of all types of all logic variables occurring in the
unification problem.
\item
The total size of all contexts of the logic variables occurring in the
unification problem.
\item
The total size of all terms in the unification problem.
\end{enumerate}
We see that the decomposition rules \textbf{dec-*} decrease (3) while
keeping (1) and (2) constant.  The lowering rules \textbf{lower-*} and
\textbf{instantiation} decrease (1).  The \textbf{intersection-eq} rule decreases
(3) while keeping (1) and (2) constant.  The \textbf{pruning},
\textbf{ctx-pruning}, and
\textbf{intersection} rules decrease (2) while keeping (1) constant.
\end{proof}

\section{Linearity  pruning}
Within the pattern fragment we know that most general unifiers exist
and we have a decidable algorithm for finding them.  For practical
applications, however, it is often necessary to relax the pattern
restriction and accept that the algorithm sometimes returns left-over
unification problems.  Reed~\cite{Reed09lfmtp}, for example, describes
the dynamic intuitionistic pattern fragment that postpones any
unification equation as constraints that cannot be solved immediately.

In this section we will relax the restriction of pattern substitutions
from Definition~\ref{def:patsub} to \emph{linear-changing pattern substitutions}
permitting linear-changing extensions, greatly
expanding the applicability of our unification algorithm. If a
unification equation involving linear-changing pattern
substitutions cannot be resolved, it is simply postponed as a
constraint.  Instead of just returning $\mathbb{T}$ or $\mathbb{F}$,
the unification algorithm using linearity pruning may fail with 
leftover constraints.


In order to handle linear-changing extensions in substitutions we
first need to revisit the notion of variable occurrence that was defined in
section~\ref{sec:inversion}.  So far, occurrences have been
divided into two categories; rigid and flexible.  We will need to make
further distinctions into a total of 12 categories.

We say that an occurrence is in an \emph{intuitionistic position}
in a term if the term can be written as $M\{n^f\cdot S\cdot (N;S')\}$ such
that the occurrence is within $N$.  If an occurrence is not in an
intuitionistic position and the term can be written as
$M\{n^f\cdot S\cdot (N\affext S')\}$ such that the occurrence is within
$N$ we say that it is in an \emph{affine position}.
If an occurrence is neither in an intuitionistic position nor in an
affine position we say that it is in a \emph{linear position}.
This means that intuitionistic positions are precisely those
in which top-level affine and linear assumptions are
not available. Similarly,  affine positions are those in which top-level affine
assumptions are available but the linear are not.  Finally,  linear
positions are those where all top-level assumptions are available.

If  $n$ occurs flexibly in a term $M$, i.e.\ it occurs in the $i$th
argument of some logic variable $X$, there are five possibilities for
the $i$th assumption in $\Gamma_X$; it can be intuitionistic, affine,
used affine, linear, or used linear.  We say that $n$ occurs in an
\emph{intuitionistic argument} if the $i$th assumption in $\Gamma_X$ is
intuitionistic, we say that it occurs in an \emph{affine argument} if
the $i$th assumption in $\Gamma_X$ is affine, and we say that it occurs
in a \emph{linear argument} if the $i$th assumption in $\Gamma_X$ is
linear.  We will write this as $n\in_{\textsf{flex},\fI}M$,
$n\in_{\textsf{flex},\fA}M$, and $n\in_{\textsf{flex},\fL}M$, respectively.
Occurrences where the $i$th assumption in $\Gamma_X$ is either
used affine or used linear are not relevant, since context pruning
will have removed them
(see rule \textbf{ctx-pruning} in Figure~\ref{fig:pat-unif}).

This gives a total of 12 categories of occurrence, since any occurrence
is either in an intuitionistic, affine, or linear position and it is
either a rigid occurrence or a flexible occurrence in an intuitionistic,
affine, or linear argument.

\begin{figure}
\[
{\newcommand{\spacer}{\rule{0pt}{3.1ex}}
\newcommand{\width}{.8\textwidth}
\begin{array}{llll}
\textbf{pruning-fail} &
P\land X[s] \doteq M
&\mapsto& \mathbb{F} \\
&\multicolumn{3}{p{\width}}{
if $n\notin s$ and either $n\in_{\textsf{rig}} M$ or $n\in_{\textsf{flex},\fL} M$
} \\
\spacer\textbf{pruning} &
P\land X[s] \doteq M
&\mapsto& [Y\leftarrow Z[w]](P\land X[s]\doteq M) \\
&\multicolumn{3}{p{\width}}{
if $n\notin s$, $n$ occurs flexibly in $M$ in the $i$th argument of
the logic variable $Y$, $w=\textsf{weaken}(\Gamma_Y;i)$ exists,
and $Z$ is a fresh logic variable with
$A_Z=A_Y$ and $\Gamma_Z=\Gamma_Y\div i$
} \\
\end{array}
}
\]
\caption{Modified pruning rules\label{fig:pat-unif-pruning}}
\end{figure} 

If we are at any time forced to prune a variable occurring in a
linear argument we can simply fail, since the reason for pruning implies
that the variable cannot occur in the given place but
the linear typing tells us that it will.
Consider the case $X[s]\doteq M$ with $n\notin s$ and $n\in M$.  Since
we have widened the fragment we are considering to include
linear-changing pattern substitutions it is now possible that
$n\in_{\textsf{flex},\fL}M$.  This was previously impossible since if
every substitution is a pattern then $n\in_{\textsf{flex},\fL}M$ implies
that $n$ is linear which in turn implies $n\in s$.
The \textbf{pruning} and \textbf{pruning-fail}
rules therefore has to be modified slightly in this case as shown in
Figure~\ref{fig:pat-unif-pruning}.

\subsection{Linear-changing pattern substitutions}
\begin{defin}
A linear-changing pattern substitution $s$ is called a
\emph{linear-changing identity} substitution if it is on the form:
\[
1^{f_1f'_1}.2^{f_2f'_2}\ldots n^{f_nf'_n}.\uparrow^n
\]
or equivalently that it is $\eta$-equivalent to \textsf{id} except for
some number of linear-changing extensions.
\end{defin}

\begin{thm}\label{thm:injective}
Linear-changing identity substitutions are injective.
Given $M$, $M'$, and a linear-changing identity substitution $s$,
then $M[s]=M'[s]$ implies $M=M'$.
\end{thm}
\begin{proof}
The substitution $s$ simply changes the linearity flags in $M$ and $M'$ from
\fL{} to \fA{} or \fI{} or from \fA{} to \fI{} on
those variables that are linear-changing in
$s$ and it is therefore trivially injective.
\end{proof}

\begin{thm}\label{thm:decompose}
A linear-changing pattern substitution can be decomposed into a pattern
substitution and a linear-changing identity substitution.
If $s$ is a linear-changing pattern substitution then there exists a
pattern substitution $s'$ and a linear-changing identity substitution
$t$ such that $s=s'\circ t$.
\end{thm}
\begin{proof}
Take $s'$ to be $s$ with all linear-changing extensions $\fA\fL$ and
$\fI\fL$ changed to
linear extensions and all linear-changing extensions $\fI\fA$ changed to
affine extensions and $t$ to be a linear-changing identity substitution
with the corresponding linear-changing extensions.
\end{proof}

\begin{thm}\label{thm:linprune}
Let $s$ be a linear-changing identity substitution with exactly one
linear-changing extension $n^{ff'}$ and $M$ be some term.
\begin{enumerate}
\item
If the linear-changing extension is $ff'=\fI\fL$
then there exists an $M'$ such that $M=M'[s]$ if
and only if the following five properties hold:
\begin{enumerate}
\item
$n$ occurs in $M$.
\item
There are no occurrences of $n$ in intuitionistic or affine positions in $M$.
\item
For all subterms $\langle M_1,M_2\rangle$ of $M$ under $k$ lambdas
$n+k$ occurs in $M_1$ if and only if it occurs in $M_2$.
\item
For all subterms $M_1\lhat{}M_2$ of $M$ under $k$ lambdas
$n+k$ occurs in at most one of $M_1$ and $M_2$.
\item
All flexible occurrences of $n$ in $M$ are in linear arguments.
\end{enumerate}
\item
If the linear-changing extension is $ff'=\fI\fA$
then there exists an $M'$ such that $M=M'[s]$ if
and only if the following three properties hold:
\begin{enumerate}
\item
There are no occurrences of $n$ in intuitionistic positions in $M$.
\item
For all subterms $M_1\lhat{}M_2$ and $M_1\affapp M_2$
of $M$ under $k$ lambdas
$n+k$ occurs in at most one of $M_1$ and $M_2$.
\item
All flexible occurrences of $n$ in $M$ are in linear or affine arguments.
\end{enumerate}
\item
If the linear-changing extension is $ff'=\fA\fL$
then there exists an $M'$ such that $M=M'[s]$ if
and only if the following four properties hold:
\begin{enumerate}
\item
$n$ occurs in $M$.
\item
There are no occurrences of $n$ in affine positions in $M$.
\item
For all subterms $\langle M_1,M_2\rangle$ of $M$ under $k$ lambdas
$n+k$ occurs in $M_1$ if and only if it occurs in $M_2$.
\item
All flexible occurrences of $n$ in $M$ are in linear arguments.
\end{enumerate}
\end{enumerate}
\end{thm}
\begin{proof}
By induction on $M$ noting that each of the three sets of
properties are precisely the occurrence requirements for, respectively,
linear variables, affine variables, and linear variables known to adhere
to the affine occurrence requirements.
\end{proof}

Theorem~\ref{thm:linprune} tells us when there exists an $M'$ such that
$M=M'[s]$ for a linear-changing identity substitution $s$ with a
single linear-changing extension.  As a corollary we get the conditions
when $s$ is a general linear-changing identity substitution.  The
existence of $M'$ is equivalent to the
conjunction of the requirements for each linear-changing extension,
since we can decompose any linear-changing identity substitution $s$
with $k$ linear-changing extensions into
$s=s_1\circ s_2\circ\dots\circ s_k$ where each $s_i$ is a
linear-changing identity substitution with exactly one linear-changing
extension.

\subsection{Linearity pruning}
Consider the following unification equation where $s$ is a
linear-changing pattern substitution:
\[
\Gamma\vdash X[s]\doteq M : B
\]
We cannot invert $s$ directly but we can decompose it by
Theorem~\ref{thm:decompose} into a pattern
substitution $s'$ and a linear-changing identity substitution $t$
changing the problem to:
\[
\Gamma\vdash X[s'][t]\doteq M : B
\]
In this case we perform a number of pruning steps on the right-hand side
since in any solution the $M$ must adhere to the requirements in
Theorem~\ref{thm:linprune}.  We will consider each linear-changing
extension $n^{ff'}$ in $t$ individually.
The entire algorithm is given in Figure~\ref{fig:lin-prun1}
and each rule is explained below.

\begin{figure}
\[
{\newcommand{\hsp}{\hspace{0ex}}
\newcommand{\spacer}{\hsp\rule{0pt}{3.1ex}}
\newcommand{\width}{.84\textwidth}
\begin{array}{llll}
\hsp\textbf{int-pos} &
P\land X[s] \doteq M
&\mapsto& P\land X[s] \doteq M\land\textsf{prune}(n+k; N) \\
&\multicolumn{3}{p{\width}}{
if $n^{ff'}$ is a linear-changing extension in $s$ and $n$ occurs in an
intuitionistic position in $M$ in the subterm $N$ under $k$ lambdas
} \\
\spacer\textbf{aff-pos} &
P\land X[s] \doteq M
&\mapsto& P\land X[s] \doteq M\land\textsf{prune}(n+k; N) \\
&\multicolumn{3}{p{\width}}{
if $n^{f\fL}$ is a linear-changing extension in $s$ and $n$ occurs in an
affine position in $M$ in the subterm $N$ under $k$ lambdas
} \\
\spacer\textbf{prune-fail} &
P\land \textsf{prune}(n; M)
&\mapsto& \mathbb{F} \\
&\multicolumn{3}{p{\width}}{
if $n\in_{\textsf{rig}}M$ or $n\in_{\textsf{flex},\fL}M$
} \\
\spacer\textbf{prune} &
P\land \textsf{prune}(n; M)
&\mapsto& [Y\leftarrow Z[w]](P\land \textsf{prune}(n; M)) \\
&\multicolumn{3}{p{\width}}{
if $n$ occurs in the
$i$th argument of the logic variable $Y$ in $M$, the argument is either
intuitionistic or affine, $w=\textsf{weaken}(\Gamma_Y;i)$,
and $Z$ is a fresh logic variable with
$A_Z=A_Y$ and $\Gamma_Z=\Gamma_Y\div i$
} \\
\spacer\textbf{prune-finish} &
P\land \textsf{prune}(n; M)
&\mapsto& P \\
&\multicolumn{3}{p{\width}}{
if $n\notin M$
} \\
\spacer\textbf{multiplicative} &
P\land X[s] \doteq M
&\mapsto& P\land X[s] \doteq M\land\textsf{prune}(n+k; M_2) \\
&\multicolumn{3}{p{\width}}{
if $n^{\fI{}f'}$ is a linear-changing extension in $s$, $n+k$ occurs either
rigidly or flexibly in a linear argument in $M_1$, and $n+k$ occurs in $M_2$, where
either $M_1\lhat{}M_2$, $M_2\lhat{}M_1$, $M_1\affapp M_2$, or $M_2\affapp M_1$
is a subterm of $M$ beneath $k$ lambdas
} \\
\spacer\textbf{additive} &
P\land X[s] \doteq M
&\mapsto& P\land X[s] \doteq M\land\textsf{prune}(n+k; M_2) \\
&\multicolumn{3}{p{\width}}{
if $n^{f\fL}$ is a linear-changing extension in $s$,
$n+k\notin M_1$, and $n+k\in M_2$, where $\langle M_1,M_2\rangle$ or
$\langle M_1,M_2\rangle$ is a subterm of $M$ beneath $k$ lambdas
} \\
\spacer\textbf{int-strengthen} &
P\land X[s] \doteq M
&\mapsto& [Y\leftarrow Z[t]](P\land X[s]\doteq M) \\
&\multicolumn{3}{p{\width}}{
if $n^{\fI{}f'}$ is a linear-changing extension in $s$, $n$ occurs
flexibly in $M$ in the $i$th argument of the logic variable $Y$, the
argument is intuitionistic, $t=\textsf{linweaken}(i;\fI\fA)$, and $Z$ is
a fresh logic variable with $A_Z=A_Y$ and
$\Gamma_Z=\textsf{strengthen}(\Gamma_Y;i;\fI\fA)$
} \\
\spacer\textbf{aff-strengthen} &
P\land X[s] \doteq M
&\mapsto& [Y\leftarrow Z[t]](P\land X[s]\doteq M) \\
&\multicolumn{3}{p{\width}}{
if $n^{\fA\fL}$ is a linear-changing extension in $s$, $n$ occurs
flexibly in $M$ in the $i$th argument of the logic variable $Y$, the
argument is affine, $t=\textsf{linweaken}(i;\fA\fL)$, and $Z$ is
a fresh logic variable with $A_Z=A_Y$ and
$\Gamma_Z=\textsf{strengthen}(\Gamma_Y;i;\fA\fL)$
} \\
\spacer\textbf{no-occur} &
P\land X[s] \doteq M
&\mapsto& \mathbb{F} \\
&\multicolumn{3}{p{\width}}{
if $n^{f\fL}$ is a linear-changing extension in $s$ and $n\notin M$
} \\
\spacer\textbf{int-aff-invert} &
P\land X[s] \doteq M
&\mapsto& P\land X[s']\doteq M' \\
&\multicolumn{3}{p{\width}}{
if $n^{\fI{}f'}$ is a linear-changing extension in $s$,
there are no occurrences of $n$ in intuitionistic positions in $M$,
for all subterms $M_1\lhat{}M_2$ and $M_1\affapp M_2$
of $M$ under $k$ lambdas
$n+k$ occurs in at most one of $M_1$ and $M_2$, and
all flexible occurrences of $n$ in $M$ are in linear or affine arguments;
$s'$ and $M'$ are given by $s=s'\circ t$ and $M=M'[t]$
where $t=\textsf{linweaken}(n;\fI\fA)$
} \\
\spacer\textbf{aff-lin-invert} &
P\land X[s] \doteq M
&\mapsto& P\land X[s']\doteq M' \\
&\multicolumn{3}{p{\width}}{
if $n^{\fA\fL}$ is a linear-changing extension in $s$,
$n$ occurs in $M$,
there are no occurrences of $n$ in affine positions in $M$,
for all subterms $\langle M_1,M_2\rangle$ of $M$ under $k$ lambdas
$n+k$ occurs in $M_1$ if and only if it occurs in $M_2$, and
all flexible occurrences of $n$ in $M$ are in linear arguments;
$s'$ and $M'$ are given by $s=s'\circ t$ and $M=M'[t]$
where $t=\textsf{linweaken}(n;\fA\fL)$
} \\
\end{array}
}
\]
\caption{Linearity pruning\label{fig:lin-prun1}}
\end{figure} 


Since many of the rules rely on pruning, we extend our language of
unification problems with the constraint $\textsf{prune}(n;M)$ to
simplify the presentation.  This
constraint states that $n$ cannot occur in $M$ in a solution.  If this is
already the case then the rule \textbf{prune-finish} removes it.  If $n$ occurs either rigidly or
flexibly in a linear argument in $M$ then no instantiation of logic
variables can remove the occurrence, and therefore there are no
solutions.  The rule \textbf{prune-fail} covers this case.  If there are
flexible occurrences in either intuitionistic or affine arguments then
we can safely prune them away with the rule \textbf{prune}.

\medskip \noindent\textbf{Position-based pruning.}
The variable $n$ cannot occur in any intuitionistic position.  
Furthermore, if $f'=\fL$ then $n$ also cannot occur in affine positions.
These occurrences can therefore be pruned away with the rules
\textbf{int-pos} and \textbf{aff-pos}.

\medskip \noindent\textbf{Pruning at multiplicative context splits.}
We will now consider all linear applications $M_1\lhat{}M_2$ and all affine
applications $M_1\affapp M_2$
in the term $M$ and compare occurrences in $M_1$ and $M_2$,
as these positions are where the context is split
multiplicatively.

For any multiplicative context split the variable should only occur in one
of the branches by Theorem~\ref{thm:linprune}.  A multiplicative split
with rigid or linear argument
occurrences in one of the branches therefore allows us to prune any
occurrences in the other branch with the rule \textbf{multiplicative},
and if this is impossible due to rigid
or linear argument occurrences in both branches, we conclude that
there is no solution by following up with \textbf{prune-fail}.
We can restrict the \textbf{multiplicative} rule to the case where
$f=\fI$, since $ff'=\fA\fL$ implies that $n$ already occurs in at most
one of the branches at each multiplicative split.

\medskip \noindent\textbf{Pruning at additive context splits.}
Similarly, we consider all pairs $\langle M_1,M_2\rangle$ in the term
$M$, i.e.\ the places where the context is split additively.
If $f'=\fL$ then the variable $n$ must occur in either both
branches of the additive split or in none of them.  An additive split without occurrences in
one of the branches therefore allows us to prune any occurrences in the
other branch using the \textbf{additive} rule.

\medskip \noindent\textbf{Strengthening intuitionistic variables.}
Consider the case when $f=\fI$, i.e.\ $n$ is intuitionistic, and
consider some flexible occurrence of $n$ in an intuitionistic argument,
say the $i$th, of some logic variable $Y$ in $M$.  If $f'=\fL$ then we
do not necessarily know whether this particular occurrence should be
pruned away or strengthened to a linear occurrence, but in either case,
and also if $f'=\fA$, we can safely
strengthen the $i$th assumption of $Y$ from intuitionistic to affine.
Let $t=\textsf{linweaken}(\Gamma_Y;i;\fI\fA)$ and $Z$ be a fresh logic variable
with $A_Z=A_Y$ and $\Gamma_Z=\textsf{strengthen}(\Gamma_Y;i;\fI\fA)$,
where $\textsf{linweaken}$ and $\textsf{strengthen}$ are defined as
follows:
\[
\begin{array}{llll}
\textsf{linweaken}(\Gamma,A^f;1;ff') &=& 1^{ff'} \\
\textsf{linweaken}(\Gamma,A^\fI;i+1;ff') &=&
1^{\fI\fI}.\textsf{linweaken}(\Gamma;i;ff')\;\circ \uparrow \\
\textsf{linweaken}(\Gamma,A^l;i+1;ff') &=&
1^{\fA\fA}.\textsf{linweaken}(\Gamma;i;ff')\;\circ \uparrow &
\textrm{if }l\in\{\fA,\fUA\} \\
\textsf{linweaken}(\Gamma,A^l;i+1;ff') &=&
1^{\fL\fL}.\textsf{linweaken}(\Gamma;i;ff')\;\circ \uparrow &
\textrm{if }l\in\{\fL,\fUL\}
\\[1.5ex]
\textsf{strengthen}(\Gamma,A^f;1;ff') &=& \Gamma,A^{f'} \\
\textsf{strengthen}(\Gamma,A^l;i+1;ff') &=& \textsf{strengthen}(\Gamma;i;ff'),A^l \\
\end{array}
\]
Note that $\Gamma\vdash\textsf{linweaken}(\Gamma;i;ff') :
\textsf{strengthen}(\Gamma;i;ff')$ when the $i$th assumption in $\Gamma$
is $A^f$ and $ff'$ is either $\fI\fA$, $\fI\fL$, or $\fA\fL$.  When
referring to $\textsf{linweaken}$ we will sometimes leave out the context
and simply write $\textsf{linweaken}(i;ff')$ as $\Gamma$ can be
inferred from the codomain of the substitution.

We can now instantiate $Y$ to $Z[t]$ as shown in the
\textbf{int-strengthen} rule.
When we cannot apply this rule anymore, and we furthermore cannot apply
any of the pruning steps described above, then either $M$ satisfies the
three conditions of part~2 of Theorem~\ref{thm:linprune} or else there
is some subterm $M_1\lhat{}M_2$ or $M_1\affapp{}M_2$ with flexible
occurrences in both $M_1$ and $M_2$.  In the latter case there is really
nothing else to do.\footnote{If we instead of a most general unifier
  were looking for the set of most general unifiers then we could
  easily enumerate the different possible solutions by introducing a
  disjunction and then either prune the variable from $M_1$ or $M_2$.}
In the former case, we can write the equation $X[s]\doteq M$ as
$X[s'][t]\doteq M'[t]$ where
$t=\textsf{linweaken}(n;\fI\fA)$.  Since $t$ is injective this equation
simplifies to $X[s']\doteq M'$, which corresponds to changing
every occurrence of $n^\fI$ to $n^\fA$.  This is
summarized by the rule \textbf{int-aff-invert}.

\medskip \noindent\textbf{Strengthening affine variables.}
Consider now the case when $ff'=\fA\fL$, i.e.\ $n$ is affine.  Since we
know that $n$ occurs affinely but should occur linearly,
no more pruning will be necessary.  This means that any
flexible occurrence of $n$ in an affine argument,
say the $i$th, of some logic variable $Y$ in $M$ can be strengthened to
a linear occurrence.  Thus, as is summarized in the
\textbf{aff-strengthen} rule we instantiate $Y$ to $Z[t]$, where $Z$ is
a fresh logic variable with $A_Z=A_Y$,
$\Gamma_Z=\textsf{strengthen}(\Gamma_Y;i;\fA\fL)$, and
$t=\textsf{linweaken}(\Gamma_Y;i;\fA\fL)$.
Since we know that $n$ is supposed to be linear then it should also occur.
If it does not, we can fail with the rule \textbf{no-occur}.

If none of the rules \textbf{no-occur}, \textbf{aff-pos},
\textbf{additive}, or \textbf{aff-strengthen} apply then $n^{\fA\fL}$
satisfies the four properties of part~3 of Theorem~\ref{thm:linprune}
and can be strengthened from affine to linear using
\textbf{aff-lin-invert}.

\bigskip

As an example we sketch how the algorithm solves
equation~\eqref{eq:split1} supposing that it has already been lowered.
$F_1[1^{\fI\fL}.\uparrow]\doteq c\lhat{}H_1[1^{\fI\fI}.\uparrow]\mapsto
F_1[1^{\fI\fL}.\uparrow]\doteq c\lhat{}H_2[1^{\fI\fA}.\uparrow]\mapsto
F_1[1^{\fA\fL}.\uparrow]\doteq c\lhat{}H_2[1^{\fA\fA}.\uparrow]\mapsto
F_1[1^{\fA\fL}.\uparrow]\doteq c\lhat{}H_3[1^{\fA\fL}.\uparrow]\mapsto
F_1[1^{\fL\fL}.\uparrow]\doteq c\lhat{}H_3[1^{\fL\fL}.\uparrow]$.  The
last equation is a pattern, which can be solved directly.

\subsection{Correctness}
The discussion above relies heavily on
Theorem~\ref{thm:linprune} and proves that the algorithm preserves
solutions.   It is therefore easily possible to generalize part~1 of the Correctness
Theorem~\ref{thm:unif-correct} to the version of the unification
algorithm including linearity pruning.
Termination (part~3) also holds for the extended algorithm with a slight
elaboration of the termination ordering.  When calculating the size of
a term we will order the linearity flags $\fI>\fA>\fL$ because with
this ordering,
the strengthening rules \textbf{int-strengthen},
\textbf{aff-strengthen}, \textbf{int-aff-invert}, and
\textbf{aff-lin-invert} decrease unification problems in size.
Furthermore, we require
that every introduction of the $\textsf{prune}(\cdot;\cdot)$ constraint is followed by
a sequence of \textbf{prune} steps followed by a \textbf{prune-fail} or
\textbf{prune-finish} step.  When the introduction and elimination of
the $\textsf{prune}(\cdot;\cdot)$ constraint are seen together as one step then the
combined result  always reduces the termination measure.
However, since the extended algorithm can get
stuck on certain equations with a ``don't know'', we have to
accept that progress, as stated in part~2 of the theorem, no
longer holds.  In these cases we can simply report a set of leftover
constraints, each of which require strengthening of some intuitionistic
variable that occurs flexibly in multiple parts of the right-hand side.

\section{Conclusion}
We have defined the pattern fragment for higher-order unification problems
in linear and affine type theory.  We have proved that all higher-order
unification equations within this fragment have no solutions or a most
general unifier, and given an algorithm to construct it.  Furthermore,
we have extended the unification algorithm beyond the pattern fragment
to those non-pattern equations that arise due to the additional
constraints from the linear and affine type system.

\bibliographystyle{eptcs}
\bibliography{lin-uni}

\end{document}